\documentclass[12pt]{article}
\usepackage{amsmath,amssymb,amsthm,bbm,natbib,color,hyperref}
\usepackage[margin=1.15in]{geometry}
\usepackage{caption}
\usepackage{graphicx}
\usepackage{subcaption}
\usepackage{afterpage}
\usepackage{microtype}
\usepackage[hang,flushmargin]{footmisc}

\newtheorem{theorem}{Theorem}

\newtheorem{example}{Example}

\newtheorem{corollary}{Corollary}

\newcommand{\R}{\mathbb{R}}

\newcommand{\EE}[1]{\mathbb{E}\left[{#1}\right]}

\newcommand{\PP}[1]{\mathbb{P}\left\{{#1}\right\}}

\newcommand{\Ppst}[3]{\mathbb{P}_{{#1}}\left\{{#2}\  \middle| \ {#3}\right\}}
\newcommand{\Pp}[2]{\mathbb{P}_{{#1}}\left\{{#2}\right\}}
\newcommand{\eqd}{\stackrel{\textnormal{d}}{=}}
\newcommand{\One}[1]{{\mathbbm{1}}\left\{{#1}\right\}}

\newcommand{\iidsim}{\stackrel{\textnormal{iid}}{\sim}}

  \newcommand\independent{\protect\mathpalette{\protect\independenT}{\perp}}
\def\independenT#1#2{\mathrel{\rlap{$#1#2$}\mkern2mu{#1#2}}}

\newcommand{\Zcal}{\mathcal{Z}}
\newcommand{\Xcal}{\mathcal{X}}
\newcommand{\Scal}{\mathcal{S}}

\usepackage{tikz}
\usetikzlibrary{decorations.pathreplacing,arrows}

\title{Permutation tests using\\ arbitrary permutation distributions}
\author{
 Aaditya Ramdas\thanks{Departments of Statistics and Machine Learning, Carnegie Mellon University} , Rina Foygel Barber\thanks{Department of Statistics, University of Chicago} ,\\ Emmanuel J.~Cand{\`e}s\thanks{Departments of Statistics and Mathematics, Stanford University} ,
Ryan J.~Tibshirani\thanks{Department of Statistics, University of California Berkeley}}

\begin{document}

\maketitle

\begin{abstract}
  Permutation tests date back nearly a century to Fisher's randomized experiments, 
  and remain an immensely popular statistical tool, used
  for testing hypotheses of independence between variables and other
  common inferential questions. 
  Much of the existing literature has emphasized that, for the permutation p-value to be valid, 
  one must first pick a subgroup $G$ of permutations (which could equal the full group)
  and then recalculate the test statistic on permuted data using either an exhaustive enumeration of $G$, 
  or a sample from $G$  drawn uniformly at random.
In this work, we demonstrate that 
  the focus on subgroups and uniform sampling are both unnecessary for validity---in fact,
 a simple random modification of the permutation p-value
  remains valid even when using an arbitrary distribution (not necessarily uniform) over any subset of permutations
(not necessarily a subgroup). We provide a unified theoretical treatment of such generalized permutation tests, recovering all known results from the literature as special cases. 
  Thus, this work expands the flexibility of the permutation test toolkit available to the practitioner.
\end{abstract}

\section{Introduction}
Suppose we observe data $X_1,\dots,X_n\in\Xcal$, and would like to test the null hypothesis
\begin{equation}\label{eq:null-exch}
H_0: \
\textnormal{$X_1,\dots,X_n$ are exchangeable}.
\end{equation}
(Note that the hypothesis that the $X_i$'s are i.i.d., is a special
case of this null.) We assume that we have a pre-specified test
statistic, which is a function $T: \Xcal^n \rightarrow \R$, where,
without loss of generality, we let larger values of
$T(X) = T(X_1,\dots,X_n)$ indicate evidence in favor of an alternative
hypothesis.

Since the null distribution of the $X_i$'s is not specified exactly,
we usually do not know the null distribution of
$T(X)$. The {\em permutation test} avoids this difficulty
by comparing $T(X)$ against the same function applied to permutations
of the data. To elaborate, let $\Scal_n$ denote the set of all
permutations on $[n]:=\{1,\dots,n\}$, and define
\begin{equation}
x_\sigma := (x_{\sigma(1)},\dots,x_{\sigma(n)})\nonumber
\end{equation}
for any $x\in\Xcal^n$ and any $\sigma\in \Scal_n$.
Then, we can compute a p-value
\begin{equation}\label{eqn:pvalue_Sn}
P = \frac{\sum_{\sigma \in \Scal_n} \One{T(X_\sigma) \geq T(X)}}{n!},
\end{equation}
which ranks $T(X)$ amongst $\{T(X_\sigma)\}_{\sigma \in \Scal_n}$ sorted in decreasing order.
 Then, under the null hypothesis $H_0$, $P$ is a valid p-value, meaning $\Pp{H_0}{P\leq \alpha}\leq \alpha$ for all
$\alpha\in[0,1]$.\footnote{Note
 that we always have $P>0$, because of the identity permutation 
 $\sigma = \mathrm{Id} \in\Scal_n$ (for which $X_\sigma = X$ and thus $\One{T(X_\sigma)\geq T(X)} = 1$). 
Other permutation p-values in this paper, like~\eqref{eq:p-random}, may explicitly include a $``1+"$ term in the numerator and denominator, but their similarity to the above formula can be intuitively justified by  thinking of the extra $``1+"$ as resulting from the identity permutation.}

As an example, suppose that the observed data set actually consists of
pairs $(X_i,Y_i)$, which are assumed to be i.i.d.~from some joint
distribution. If we are interested in testing
$H'_0: X\independent Y$, we can reframe this question as testing whether
$X_1,\dots,X_n$ are i.i.d.~conditional on $Y_1,\dots,Y_n$---in particular, under $H_0'$,
it holds that $X$ follows an exchangeable distribution {\em conditional} on $Y$.  Our test
statistic $T$ might be chosen as
\[T(X) = \left|\textnormal{Corr}\big((X_1,\dots,X_n),(Y_1,\dots,Y_n)\big)\right|.
\]
In order to see whether the observed correlation is sufficiently large to be statistically significant, we would compare $T(X,Y)$ to
the correlations computed on permuted data, 
\[T(X_\sigma)=
\left|\textnormal{Corr}\big((X_{\sigma(1)},\dots,X_{\sigma(n)}),(Y_1,\dots,Y_n)\big)\right|.
\]
The resulting p-value computed as in~\eqref{eqn:pvalue_Sn} is then a valid p-value under the null hypothesis $H_0'$.
In addition to testing independence, permutation tests are also
commonly used for testing other hypotheses, such as whether two samples
follow the same distribution.\footnote{Permutation tests are a special case of
``invariance-based testing''
\citep[Chapter 6]{lehmann2005testing}.}

The p-value $P$ computed in~\eqref{eqn:pvalue_Sn} requires computing
$T(X_\sigma)$ for every $\sigma\in\Scal_n$. One may naturally be
interested in reducing the computational cost of this procedure, since
computing $T(X_\sigma)$ for $|\Scal_n|=n!$ many permutation may be
computationally prohibitive for even moderately large $n$. As is well
known, we can obtain valid p-values by \textit{uniformly} randomly sampling permutations
from $\Scal_n$ and computing 
\begin{equation}\label{eq:p-random}
P = \frac{1 + \sum_{m=1}^M \One{T(X_{\sigma_m})\geq T(X)}}{1+M}, 
\end{equation}
in which $\sigma_1, \ldots, \sigma_M$ are i.i.d.~uniform draws from $\Scal_n$.

In a different direction, one can also reduce the set of permutations
$\sigma$ to subsets of $\Scal_n$. Specifically, let
$G\subseteq \Scal_n$ be any subset, and define
\begin{equation}\label{eqn:pvalue_subgroup}P = \frac{\sum_{\sigma \in G} \One{T(X_\sigma)\geq T(X)}}{|G|},\end{equation}
where $|G|$ is the cardinality of $G$. 

Clearly, if $G$ does not contain the identity permutation, then $P$ cannot be a p-value because it could
potentially take on the value zero. However, including the identity permutation is not sufficient.
The literature repeatedly emphasizes that $P$ defined
in~\eqref{eqn:pvalue_subgroup} is a valid p-value only if the subset 
$G\subseteq\Scal_n$ is in fact a \emph{subgroup}\footnote{For
  completeness, a group is a set paired with an operation that takes any two  
  elements of the set and produces a third, such that the operation is
  associative, an identity element exists, and every element has an inverse. A
  subgroup is just a subset of the original group that maintains the same
  properties---in particular, any subgroup must contain the identity
    element. The group $\Scal_n$ is called the symmetric group; its elements
  are the $n!$ permutations over $n$ objects. The operation is denoted $\circ$,
  sometimes called ``composition''; its action is to compose any two
  permutations $\sigma,\sigma'$ to yield a third one $\nu := \sigma \circ
  \sigma'$ which is given by $\nu(i)=\sigma(\sigma'(i))$ for $i \in [n]$. The
  inverse of $\sigma$, denoted $\sigma^{-1}$, is defined by setting
  $\sigma^{-1}(i) = j$ if $\sigma(j)=i$, so that $\sigma \circ \sigma^{-1}$
  always equals the identity permutation introduced earlier. Note that $\Scal_n$
  is not an Abelian group, meaning that $\circ$ is not commutative, since usually, $\sigma
  \circ \sigma' \neq \sigma' \circ \sigma$.
  For a subset $G\subseteq\Scal_n$, we can verify that $G$ is a
  subgroup if it is closed under composition (i.e., $\sigma\circ\sigma'\in G$
  for any $\sigma,\sigma'\in G$).}
\citep[Theorem 1]{hemerik2018exact}.

The subgroup $G$ mentioned above may be chosen
strategically to balance between computational efficiency and the
power of the test (see, e.g.,
\citet{hemerik2018exact,koning2022faster}).
In case $G$ has a large cardinality, the aforementioned references show
that sampling permutations \emph{uniformly at random} from $G$ also
yields valid p-values---that is, the randomized p-value $P$ from~\eqref{eq:p-random}
is valid if $\sigma_1,\dots,\sigma_M$ are i.i.d.~samples drawn uniformly from $G$
rather than from $\Scal_n$.  Again, choosing $G$ to be a subgroup (rather than an arbitrary subset),
 and sampling uniformly rather than from an arbitrary distribution, 
are both important for the validity of the resulting p-value.

\subsection{Contributions}

 The background above naturally leads to the following question:
while it is indeed correct that $P$ from \eqref{eqn:pvalue_subgroup} is 
not a p-value for general subsets $G \subseteq \Scal_n$,
is it possible to slightly modify the definition of $P$ so that it retains its validity
for subsets $G$ that are not subgroups? Further, while sampling the $\sigma_m$'s nonuniformly 
from $G$ would destroy the validity of $P$ from {\eqref{eq:p-random}, can we modify the definition
of $P$ so that nonuniform sampling from a set is allowed?
The first question is addressed by  \citet{hemerik2018exact}, as we will describe below; to our knowledge, the second question  
has not been addressed in the literature.

In this paper, we will broaden the applicability of
permutation tests and present generalizations, which yield valid
p-values even when we sample permutations---with or without
replacement---from a non-uniform distribution over all permutations or
from an arbitrary subset of permutations. In doing so, we shall
carefully explain how this generalization relates and extends all
previous options. 
This generalization yields new and more flexible permutation test methods; we leave a detailed study of
pros and cons of these generalizations (such as how they tradeoff the two types of errors) to future work.

\section{A generalized permutation test}\label{sec:main}

\subsection{Testing with an arbitrary distribution}
\label{sec:testing}

We now present our first generalization of the permutation test.  It allows us to use any (not necessarily uniform) distribution over $\Scal_n$ in order to construct our permutation p-value.
\begin{theorem}\label{thm:distrib_on_Sn}
  Let $q$ be any distribution over $\sigma\in\Scal_n$.  Let
  $\sigma_0\sim q$ be a random draw, and define
\begin{equation}\label{eqn:pvalue_general_distribution}
P = \sum_{\sigma\in\Scal_n} q(\sigma) \cdot\One{T(X_{\sigma\circ\sigma_0^{-1}})\geq T(X)}.\end{equation}
Then $P$ is a valid p-value, 
 i.e., 
$\Pp{H_0}{P\leq \alpha}\leq \alpha$ for all $\alpha\in[0,1]$.
\end{theorem}
In this theorem, validity is retained when conditioning on the order
statistics of the data, meaning that
$\Ppst{H_0}{P\leq \alpha}{X_{(1)},\dots,X_{(n)}}\leq \alpha$, where
$X_{(1)}\leq \dots \leq X_{(n)}$ are the order statistics of
$X=(X_1,\dots,X_n)$.\footnote{The notation of the order statistics implicitly
 assumes $\Xcal=\R$. More generally, for an arbitrary space $\Xcal$, the validity of $P$ is retained
 when conditioning on the unordered observed data, i.e., the multiset $\{X_1,\dots,X_n\}$.} The reason that this holds is simply because
$H_0$ remains true even conditional on the order statistics---that is,
if $X$ is exchangeable, then $X\mid (X_{(1)},\dots,X_{(n)})$ is again
exchangeable. The same conditional validity holds for all results to
follow, as well.  However, one cannot condition on $\sigma_0$;
the result only holds marginally over $\sigma_0$, and 
this external randomization is key to retaining validity.

We defer the proof to Section~\ref{sec:proof1}---we will first discuss
connections to the existing literature  in order to provide more context and intuition for the above theorem.

\paragraph{Uniform distribution over a subgroup.} To begin with, assume $q$
is the uniform distribution over a fixed subgroup $G$ of
$\Scal_n$. Then in this case, the p-value in
\eqref{eqn:pvalue_general_distribution} takes the special form
\[
P = \sum_{\sigma\in G} \frac{1}{|G|}
\cdot\One{T(X_{\sigma\circ\sigma_0^{-1}})\geq T(X)} = \sum_{\sigma\in G} \frac{1}{|G|}
\cdot\One{T(X_{\sigma})\geq T(X)},  
\]
where the second equality holds because a subgroup $G$ is closed under
inverses and composition, so
$\{\sigma\circ \sigma_0^{-1} : \sigma\in G\} = G$ for any
$\sigma_0\in G$. This simple observation recovers a well-known fact we
discussed earlier; namely, one can restrict the set of permutations to
an arbitrary subgroup, and the p-value $P$ defined in Theorem~\ref{thm:distrib_on_Sn} will then coincide with our earlier definition~\eqref{eqn:pvalue_subgroup}
 (proved to be a valid p-value in \citep[Theorem 1]{hemerik2018exact}).

\paragraph{Uniform distribution over a subset.}
Consider now a uniform distribution over an arbitrary subset $S$ that
is not a subgroup. In this case, the definition of $P$ in Theorem~\ref{thm:distrib_on_Sn} is equal to 
\begin{equation}\label{eqn:pvalue_subset_corrected}
P = \frac{\sum_{\sigma \in S} \One{T(X_{\sigma\circ\sigma_0^{-1}})\geq
  T(X)}}{|S|},
\end{equation}
as proposed earlier by \citet{hemerik2018exact}. 
This is, in general, not the same as
\begin{equation}\label{eqn:pvalue_subset}
P' = \frac{\sum_{\sigma \in S} \One{T(X_\sigma)\geq T(X)}}{|S|}
\end{equation}
(which is equivalent to the quantity defined in~\eqref{eqn:pvalue_subgroup} earlier,
with the subset $S$ in place of a subgroup $G$). As we shall see below, $P'$ is generally not a p-value, a fact which can cause large issues in
practice, as has been frequently emphasized.  For example, consider
the tool of {\em balanced permutations}---in the setting of testing
whether a randomly assigned treatment has a zero or nonzero effect,
this method has been proposed as a variant of the permutation test in
this setting, where the subset $S$ consists of all permutations such
that the permuted treatment group contains exactly half of the
original treatment group, and half of the original control group.
\citet{southworth2009properties} show that the quantity $P'$ computed
in~\eqref{eqn:pvalue_subset} for this choice of subset $S$ can be
substantially anti-conservative, i.e., $\PP{P\leq \alpha}> \alpha$,
particularly for low significance levels $\alpha$.  (See also
\citet{hemerik2018exact} for additional discussion of this issue.)

A simple example may help to illustrate this point, and to give intuition for the role of the random permutation $\sigma_0$.
\begin{example}\label{example1}
  Let $n=4$, and consider the set
\[
S = \{\mathrm{Id} , \sigma_{1\leftrightarrow 3, 2\leftrightarrow 4}, \sigma_{1\leftrightarrow 4,  2\leftrightarrow 3}\},
\]
where, e.g., $\sigma_{1\leftrightarrow 3, 2\leftrightarrow 4}$ is the
permutation swapping entries $1$ and $3$ and also swapping $2$ and
$4$.  Let \smash{$X_1,X_2,X_3,X_4\iidsim \mathcal{N}(0,1)$} be standard normal
random variables (so that the null hypothesis of exchangeability, $H_0$, is satisfied), and set $T(X) = X_1 + X_2$. Then the quantity $P'$
defined in~\eqref{eqn:pvalue_subset} is equal to
\[P' = \frac{\One{T(X_{\mathrm{Id}})\geq T(X)} +
  \One{T(X_{\sigma_{1\leftrightarrow 3, 2\leftrightarrow 4}})\geq
    T(X)} + \One{T(X_{\sigma_{1\leftrightarrow 4, 2\leftrightarrow
        3}})\geq T(X)}}{3}.
\]
This gives
\[
P' = \begin{cases}\frac{1+0+0}{3} = \frac{1}{3}, & \textnormal{ if }X_3
  + X_4 < X_1 + X_2 , \\ \frac{1+1+1}{3}= 1, & \textnormal{
    otherwise}\end{cases}
\]
and, therefore, 
\[P' = \begin{cases} \frac{1}{3},&\textnormal{ with probability
    $\frac{1}{2}$},\\ 1,&\textnormal{ with probability
    $\frac{1}{2}$.}\end{cases}\]
We can thus see that $P'$ is anti-conservative at the threshold
$\alpha = \frac{1}{3}$.

Next, we will see how the correction~\eqref{eqn:pvalue_subset_corrected} fixes
the failure described above. Denote by $P_{\sigma}$ the p-value
in~\eqref{eqn:pvalue_subset_corrected} calculated conditional on the random
$\sigma_0$ being equal to $\sigma$, so that 
\begin{equation}
P = \begin{cases} 
P_{\mathrm{Id}}, & \textnormal{ w.p. } 1/3, \\
P_{\sigma_{1\leftrightarrow 4, 2\leftrightarrow 3}}, & \textnormal{ w.p. } 1/3, \\
P_{\sigma_{1\leftrightarrow 3, 2\leftrightarrow 4}}, & \textnormal{ w.p. } 1/3.
\end{cases}
\end{equation}
Then, the calculation that was previously performed effectively shows that
\[
P_{\mathrm{Id}} =  \begin{cases} \frac{1}{3}, & \textnormal{ if }X_3 + X_4 < X_1 + X_2 , \\  1, & \textnormal{ otherwise}.\end{cases}
\]
A similar straightforward calculation then yields
\[
P_{\sigma_{1\leftrightarrow 3, 2\leftrightarrow 4}} = P_{\sigma_{1\leftrightarrow 4, 2\leftrightarrow 3}} = \begin{cases} \frac{2}{3}, & \textnormal{ if }X_3 + X_4 < X_1 + X_2 , \\  1, & \textnormal{ otherwise}.\end{cases}
\]
Put together, we obtain 
\begin{equation}\label{eq:example-fixed-P}
P =  \begin{cases} 
\frac{1}{3}, & \textnormal{ w.p. } 1/6, \\
\frac{2}{3}, & \textnormal{ w.p. } 1/3, \\
1, & \textnormal{ w.p. } 1/2.
\end{cases}
\end{equation}
This is indeed stochastically larger than uniform, and is thus a valid p-value.
\end{example}

\paragraph{The role of $\sigma_0$.}
To better understand the role of the random permutation $\sigma_0$, let us consider Example~\ref{example1}
again, and look more closely at what goes wrong there.
We observe that $P'$ compares the observed statistic $T(X)$ against the  set $
\{T(X_\sigma)\}_{\sigma\in S} = \{T(X_{\textnormal{Id}}),T(X_{\sigma_{1\leftrightarrow 3, 2\leftrightarrow 4}}) , T(X_{\sigma_{1\leftrightarrow 4, 2\leftrightarrow        3}})\}$. 
For $P'$ to be a valid p-value, given the (unordered) set of potential data vectors  $\{X_{\textnormal{Id}},X_{\sigma_{1\leftrightarrow 3, 2\leftrightarrow 4}} , X_{\sigma_{1\leftrightarrow 4, 2\leftrightarrow        3}}\}$,
it suffices that the actual observed data $X$ is equally likely to be any one of these three. Now suppose this set is equal
to $\{(0.8, 0.5, 0.2, 1.0), (1.0, 0.2, 0.5, 0.8), (0.5, 0.8, 1.0, 0.2)\}$, in no particular order. 
Each of these three vectors have equal likelihood under $H_0$ (due to exchangeability). 
Counterintuitively, however, our knowledge of the subset of permutations $S$ implies that
we {\em must} have $X = (1.0, 0.2, 0.5, 0.8)$---otherwise we could not have obtained this particular set. For instance, if $X=(0.8, 0.5, 0.2, 1.0)$,
then we would have $X_{\sigma_{1\leftrightarrow 3, 2\leftrightarrow 4}} = (0.2,
1.0, 0.8, 0.5)$---but this does not lie in our set, so it cannot be the
correct value of $X$. In other words, if we condition on the unordered 
set $\{X_\sigma\}_{\sigma\in S}$, which is the {\em orbit} of the data $X$ under the actions of permutations $\sigma\in S$,
our intuition tells us that $X$ is equally likely to be any element of this orbit---but in fact, for a non-subgroup $S$, $X$ might be uniquely identified
from its orbit.

Now consider what happens if we compute the corrected p-value $P$~\eqref{eqn:pvalue_subset_corrected}  and
let us once more examine the question of identifying the data $X$ from its orbit. The  p-value $P$
compares the observed statistic $T(X)$ against the  set $\{T(X_{\sigma\circ\sigma_0^{-1}})\}_{\sigma\in S}$,
and so now the question is whether we can identify $X$ from the set $\{X_{\sigma\circ\sigma^{-1}_0}\}_{\sigma\in S}$, which is the 
orbit of $X_{\sigma_0^{-1}}$ for a randomly drawn $\sigma_0\in S$. 
Identifying $X$ is no longer possible because of the random $\sigma_0$. 
For instance, working again with Example~\ref{example1}, suppose this set $\{X_{\sigma\circ\sigma^{-1}_0}\}_{\sigma\in S}$ is equal
to $\{(0.8, 0.5, 0.2, 1.0), (1.0, 0.2, 0.5, 0.8), (0.5, 0.8, 1.0, 0.2)\}$, in no particular order. 
We can identify that this is the orbit of $x = (1.0, 0.2, 0.5, 0.8)$ under $S$---that is, this set is equal to $\{x_\sigma\}_{\sigma\in S}$.
Then the following three possibilities are equally likely:
\begin{itemize}
\item $\sigma_0 = \textnormal{Id}$ and so $X = x_{ \textnormal{Id}^{-1}} = x = (1.0, 0.2, 0.5, 0.8)$;
\item $\sigma_0 = \sigma_{1\leftrightarrow 3, 2\leftrightarrow 4}$ and so $X = x_{\sigma_{1\leftrightarrow 3, 2\leftrightarrow 4}^{-1}} = (0.5,0.8,1.0,0.2)$;
\item $\sigma_0 = \sigma_{1\leftrightarrow 4, 2\leftrightarrow 3}$ and so $X = x_{\sigma_{1\leftrightarrow 4, 2\leftrightarrow 3}^{-1}} = (0.8,0.5,0.2,1.0)$.
\end{itemize}
In other words, $X$ is now equally likely to be any of the three values in our set, and validity is restored.

\subsection{Random samples from an arbitrary distribution}

Our second generalization concerns permutations that are randomly
chosen from an arbitrary distribution.
\begin{theorem}\label{thm:distrib_on_Sn_sample}
Let $q$ be any distribution over $\sigma\in\Scal_n$. 
Let \smash{$\sigma_0,\sigma_1,\dots,\sigma_M\iidsim q$}, and define
\begin{equation}\label{eqn:pvalue_general_distribution_sample}
P = \frac{1 + \sum_{m=1}^M \One{T(X_{\sigma_m\circ\sigma_0^{-1}})  \geq T(X)} }{1+M}.\end{equation}
Then $P$ is a valid p-value, 
 i.e., 
$\Pp{H_0}{P\leq \alpha}\leq \alpha$ for all $\alpha\in[0,1]$.
\end{theorem}
\noindent 
This result is closely related to \citet{besag1989generalized}'s well
known construction for obtaining exchangeable samples from Markov
chain Monte Carlo (MCMC) sampling---the details are deferred to
Section~\ref{sec:mcmc}.

Just as before, some special cases of this result are well known to
statisticians.

\paragraph{Random permutations from $\Scal_n$.} In the simple case
where $q$ is the uniform distribution over $\Scal_n$, Theorem
\ref{thm:distrib_on_Sn_sample} states that
\begin{equation}
\label{eq:easy}
 P = \frac{1 + \sum_{m=1}^M \One{T(X_{\sigma_m\circ\sigma_0^{-1}})
     \geq T(X)} }{1+M} \eqd  \frac{1 + \sum_{m=1}^M \One{T(X_{\sigma_m})
     \geq T(X)} }{1+M}
\end{equation}
is a valid p-value. The equality in distribution above holds because
the ${\sigma_m\circ\sigma_0^{-1}}$'s are i.i.d.~draws from
$\Scal_n$. Hence, this recovers the most commonly implemented form of the permutation test.

\paragraph{Random permutations from a subgroup.} The distributional
equality \eqref{eq:easy} extends to any uniform distribution $q$ over
a subgroup $G$ of $\Scal_n$ since in this case, as before, the random variables
${\sigma_m\circ\sigma_0^{-1}}$ are i.i.d.~draws from $G$. This gives
the following well-known result (see, e.g., \citet[Theorem
2]{hemerik2018exact}):
\begin{corollary}\label{thm:randomperms-subgroup}
Let $G\subseteq\Scal_n$ be a subgroup, and sample
$\sigma_1,\dots,\sigma_M\iidsim\textnormal{Unif}(G)$. Then
\begin{equation}\label{eqn:pvalue_random_subgroup}
P = \frac{1 + \sum_{m=1}^M \One{T(X_{\sigma_m})\geq T(X)}}{1+M}\end{equation}
is a valid p-value, i.e., 
$\Pp{H_0}{P\leq \alpha}\leq \alpha$ for all $\alpha\in[0,1]$.
\end{corollary}

\paragraph{Random permutations from a subset.} Consider now case where
$q$ is a uniform distribution over an arbitrary subset $S$. When $S$
is not a subgroup, the $P$ value
\[
P = \frac{1 + \sum_{m=1}^M \One{T(X_{\sigma_m\circ\sigma_0^{-1}})
     \geq T(X)} }{1+M} 
\]
may not have the same distribution as 
\[
P' = \frac{1 + \sum_{m=1}^M \One{T(X_{\sigma_m})
     \geq T(X)} }{1+M}. 
\]
Here, Theorem~\ref{thm:distrib_on_Sn_sample} gives:
\begin{corollary}\label{thm:subset_sample}
Let $S\subseteq\Scal_n$ be any fixed subset of permutations. 
Sample $\sigma_0,\sigma_1,\dots,\sigma_M\iidsim\textnormal{Unif}(S)$. 
Then
\begin{equation}\label{eqn:pvalue_random}
P = \frac{1 + \sum_{m=1}^M \One{T(X_{\sigma_m \circ \sigma_0^{-1}})\geq T(X)}}{1+M}\end{equation}
is a valid p-value, i.e., 
$\Pp{H_0}{P\leq \alpha}\leq \alpha$ for all $\alpha\in[0,1]$.
\end{corollary}
\noindent To the best of our knowledge, this statement had not been
recorded in the literature. As a variant, the same result holds if we
instead draw permutations without replacement.
\begin{corollary}\label{cor:subset_sample_without_replacement}
  Consider the variant of Corollary~\ref{thm:subset_sample} in which the
  permutations are drawn without replacement. Then the p-value $P$ defined 
in~\eqref{eqn:pvalue_random} is a valid p-value, i.e.,
  $\Pp{H_0}{P\leq \alpha}\leq \alpha$ for all $\alpha\in[0,1]$. In the
  special case where $S$ is a subgroup, the same conclusion  also applies
for the p-value $P$ defined in~\eqref{eqn:pvalue_random_subgroup}.
\end{corollary}
\begin{proof}
  Let $S'\subseteq S$ be a subset of size $M+1$ chosen uniformly at
  random. Let $\sigma_0,\sigma_1,\dots,\sigma_M$ be a random ordering
  of $S'$---in particular, this means that $\sigma_0$ is drawn
  uniformly from $S'$. Then by Theorem~\ref{thm:distrib_on_Sn},  applied
  with $q$ taken to be the uniform distribution over $S'$, $P$ is a
  valid p-value.

The second claim follows from the fact that 
\[(\sigma_1\circ\sigma_0^{-1},\dots,\sigma_M\circ\sigma_0^{-1})\eqd
(\sigma_1,\dots,\sigma_M),
\]
whenever $S$ is a
subgroup.
\end{proof}

To guide the reader, Figure~\ref{fig:flowchart} summarizes the
connections between all the results presented thus far in the paper.
Interestingly, as highlighted in the figure, Theorems~\ref{thm:distrib_on_Sn} and~\ref{thm:distrib_on_Sn_sample}
can be derived from each other; we will elaborate on this connection below.

 Finally, we present another simple example to highlight
the necessity of the $\sigma_0$ term, in the case of nonuniform sampling. 
Indeed, even  ``intuitive'' modifications of the uniform sampling scheme may fail to produce valid p-values.
\begin{example}\label{ex:anthithetic}
If one considers $P' = \frac{1 + \sum_{m=1}^M \One{T(X_{\sigma_m})
     \geq T(X)} }{1+M}$ to be a Monte Carlo estimate of
the p-value $P = \frac{\sum_{\sigma \in \Scal_n} \One{T(X_\sigma) \geq T(X)}}{n!}$ computed in~\eqref{eqn:pvalue_Sn}, then a lower variance estimate may be obtained by 
``antithetic sampling''---that is, pairing a random draw $\sigma_m\in\Scal_n$ with its reverse $\mathrm{Rev}(\sigma_m) = (\sigma_m(n),\dots,\sigma_m(1))$ (see,
e.g.,~\citet{mitchell2022sampling} for an example of this variance reduction technique). However, 
using antithetic sampling can lead to an invalid  p-value---specifically, if $\sigma_1,\dots,\sigma_{M/2}$ are drawn uniformly at random from $\Scal_n$ (or 
from some subgroup $G\subseteq\Scal_n$), and we then set $\sigma_{M/2 + i } = \mathrm{Rev}(\sigma_i)$ for each $i=1,\dots,M/2$,
then the quantity $P'$ may not be a valid p-value. For instance, suppose we take $M=2$, so that $\sigma_2 = \mathrm{Rev}(\sigma_1)$
where $\sigma_1$ is drawn uniformly from $\Scal_n$. Take $T(x) = X_1 + X_n$,
and draw $X_i\iidsim \mathcal{N}(0,1)$. Then
\begin{multline*}P' = \frac{1 + \One{X_{\sigma_1(1)} + X_{\sigma_1(n)} \geq X_1 + X_n} + \One{X_{\sigma_2(1)} + X_{\sigma_2(n)} \geq X_1 + X_n} }{3}\\
=  \frac{1 + 2\cdot \One{X_{\sigma_1(1)} + X_{\sigma_1(n)} \geq X_1 + X_n}}{3}.\end{multline*}
Then we can verify that, conditional on $\sigma_1$, if $\{\sigma_1(1),\sigma_1(n)\} = \{1,n\}$ then $P=1$,
while if $\{\sigma_1(1),\sigma_1(n)\} \neq \{1,n\}$ then $P'=\frac{1}{3}$ or $P'=1$ each with probability $\frac{1}{2}$, which yields
\[\PP{P' = \frac{1}{3}} = \frac{1}{2} - \frac{1}{n(n-1)} > \frac{1}{3},\]
with the last step holding if $n>3$. We can thus see that $P'$ is anti-conservative at the threshold
$\alpha = \frac{1}{3}$.
\end{example}

\begin{figure}[t] \centering\footnotesize
\begin{tikzpicture}

\begin{scope}[xshift=0cm,yshift=0cm]

\node (q) at (2,1.75) {\begin{tabular}{c}distrib.~$q$ on $\Scal_n$
                         \\ Theorem~\ref{thm:distrib_on_Sn} \end{tabular}};
\node (q1) at (1.8,1.75) {};
\node (q2) at (2.2,1.75) {};
\node (qsample) at (2,-1.75) {\begin{tabular}{c}samples from $q$\\Theorem~\ref{thm:distrib_on_Sn_sample} \end{tabular}};
\node (qsample1) at (1.8,-1.75) {};
\node (qsample2) at (2.2,-1.75) {};
\node (G) at (8,2.5) {\begin{tabular}{c}subgroup
                        $G\subseteq\Scal_n$~\eqref{eqn:pvalue_subgroup}\\{[well known]} \end{tabular}};
\node (S) at (8,1) {\begin{tabular}{c}subset $S\subseteq\Scal_n$~\eqref{eqn:pvalue_subset_corrected}\\ \citep{hemerik2018exact} \end{tabular}};
\node (S1) at (9,0.5) {};
\node (Gsample) at (8,-1) {\begin{tabular}{c}sample from $G$ w/
                             repl.\\ Corollary~\ref{thm:randomperms-subgroup}\\\citep{hemerik2018exact}\end{tabular}};
\node (Ssample) at (8,-2.5) {\begin{tabular}{c}sample from $S$ w/ repl.\\Corollary~\ref{thm:subset_sample} \\\end{tabular}};
\node (Gsamplewo) at (14,-1) {\begin{tabular}{c}sample from $G$ w/o repl.\\Corollary~\ref{cor:subset_sample_without_replacement} \end{tabular}};
\node (Ssamplewo) at (14,-2.5) {\begin{tabular}{c}sample from $S$ w/o repl.\\Corollary~\ref{cor:subset_sample_without_replacement} \end{tabular}};
\draw[->][line width=0.5mm,shorten >=15pt,shorten <=15pt] (qsample1.north) -- (q1.south);
\draw[->][line width=0.5mm,shorten >=15pt,shorten <=15pt] (q2.south) -- (qsample2.north);
\draw[->][line width=0.5mm,shorten >=5pt,shorten <=5pt] (q.east) -- (G.west);
\draw[->][line width=0.5mm,shorten >=5pt,shorten <=5pt] (q.east) -- (S.west);
\draw[->][line width=0.5mm,shorten >=5pt,shorten <=5pt] (qsample.east) -- (Gsample.west);
\draw[->][line width=0.5mm,shorten >=5pt,shorten <=5pt] (qsample.east) -- (Ssample.west);
\draw[->][line width=0.5mm,shorten >=0pt,shorten <=5pt] (S1.south east) -- (Gsamplewo.north west);
\draw[->][line width=0.5mm,shorten >=0pt,shorten <=5pt] (S1.south east) -- (Ssamplewo.north west);
\node[rotate=90] (LLN) at (1.6,0) {\scriptsize proved via LLN};
\node[rotate=270] (LLN) at (2.4,0) {\scriptsize proved via $\widehat{q}$};
\end{scope}

\end{tikzpicture}
\caption{A flowchart to illustrate the connections between (most of) the results presented in this paper. Arrows point from more general results
to their special cases.}
\label{fig:flowchart}
\end{figure}
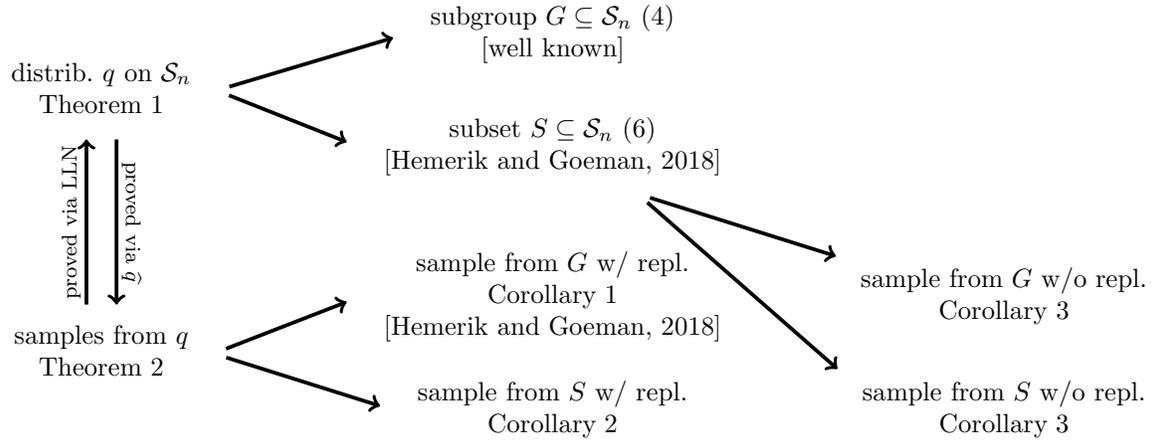

\subsection{Proof of Theorem~\ref{thm:distrib_on_Sn}}
\label{sec:proof1}

First, for any fixed $\sigma'\in \Scal_n$, we have
\begin{equation}\label{eqn:eqd_distrib} \PP{ \sum_{\sigma\in\Scal_n} q(\sigma) \cdot\One{T(X_{\sigma})\geq T(X_{\sigma'})}\leq \alpha} =  \PP{\sum_{\sigma\in\Scal_n} q(\sigma) \cdot\One{T(X_{\sigma\circ\sigma'{}^{-1}})\geq T(X)} \leq \alpha} \end{equation} because $X \eqd X_{\sigma'}$ under $H_0$ (and note 
that $(X_{\sigma'})_{\sigma\circ\sigma'{}^{-1}} = X_\sigma$).
Next, we will apply a deterministic inequality by \cite{harrison2012conservative}:
for all $t_1,\dots,t_N \in [-\infty,\infty]$ and all $\alpha,w_1,\dots,w_N \in [0,\infty]$, 
\[
\sum_{k=1}^N w_k \One{\sum_{i=1}^N w_i \One{t_i \geq t_k} \leq \alpha} \leq \alpha.
\]
Applying this bound with $q(\sigma)$'s in place of the $w_i$'s, and $T(X_\sigma)$'s in place of the $t_i$'s,
we obtain
\begin{multline}\label{eqn:harrison}
\sum_{\sigma'\in \Scal_n} q(\sigma')\cdot  \PP{ \sum_{\sigma \in \Scal_n} q(\sigma)\cdot \One{T(X_{\sigma}) \geq T(X_{\sigma'})}  \leq \alpha} \\
= 
\EE{\sum_{\sigma'\in \Scal_n} q(\sigma')\cdot  \One{\sum_{\sigma \in \Scal_n} q(\sigma)\cdot  \One{T(X_{\sigma}) \geq T(X_{\sigma'})} \leq \alpha} } 
\leq 
\alpha .\end{multline}
Finally, we have
\begin{align*}
\PP{P\leq \alpha}
&=\PP{\sum_{\sigma \in \Scal_n}q(\sigma)\cdot \One{T(X_{\sigma\circ\sigma_0^{-1}}) \geq T(X)}  \leq \alpha} \\
&=\sum_{\sigma'\in \Scal_n} \PP{\sigma_0 = \sigma' \textnormal{ and }\sum_{\sigma \in \Scal_n}q(\sigma)\cdot \One{T(X_{\sigma\circ\sigma'{}^{-1}}) \geq T(X)} \leq \alpha} \\
&=\sum_{\sigma'\in \Scal_n}q(\sigma')\cdot \PP{\sum_{\sigma \in \Scal_n}q(\sigma)\cdot \One{T(X_{\sigma\circ\sigma'{}^{-1}}) \geq T(X)} \leq \alpha} \\
&=\sum_{\sigma'\in \Scal_n}q(\sigma')\cdot \PP{\sum_{\sigma \in \Scal_n}q(\sigma)\cdot \One{T(X_{\sigma}) \geq T(X_{\sigma'})} \leq \alpha} 
\leq \alpha,\end{align*}
where the third step holds since $\sigma_0\sim q$ is drawn independently
of the data $X$, while the last two steps
apply~\eqref{eqn:eqd_distrib} and~\eqref{eqn:harrison}.

\subsection{Connecting Theorems~\ref{thm:distrib_on_Sn} and~\ref{thm:distrib_on_Sn_sample}}
As mentioned earlier, Theorems~\ref{thm:distrib_on_Sn}
and~\ref{thm:distrib_on_Sn_sample} can be derived from each other.   We
now give these proofs to show the connection.

\begin{proof}[Alternative proof of Theorem~\ref{thm:distrib_on_Sn} (via Theorem~\ref{thm:distrib_on_Sn_sample})]
Let \smash{$\sigma_0,\sigma_1,\sigma_2,\dots\iidsim q$}, and for any fixed $M$,
define
\begin{multline*}P_M = \frac{1 + \sum_{m=1}^M \One{T(X_{\sigma_m\circ\sigma_0^{-1}}) \geq T(X)}}{1+M}= \frac{ \sum_{m=0}^M \One{T(X_{\sigma_m\circ\sigma_0^{-1}}) \geq T(X)}}{1+M}\\
 = \sum_{\sigma\in \Scal_n} \frac{\sum_{m=0}^M \One{\sigma_m = \sigma}}{1+M} \cdot \One{T(X_{\sigma\circ\sigma_0^{-1}})\geq T(X)}.\end{multline*}
By the Law of Large Numbers, we see that $\frac{\sum_{m=0}^M \One{\sigma_m = \sigma}}{1+M}\rightarrow q(\sigma)$ almost surely
for all $\sigma\in\Scal_n$, and therefore, 
$P_M\rightarrow P$ almost surely, where $P$ is the p-value defined in~\eqref{eqn:pvalue_general_distribution}. In particular,
this implies that $P_M$ converges to $P$ in distribution, and therefore
\[\PP{P\leq \alpha} = \lim_{M\rightarrow\infty}\PP{P_M\leq \alpha} \leq \alpha,\]
where the last step holds since, for every $M\geq 1$, $P_M$ is a valid p-value by Theorem~\ref{thm:distrib_on_Sn_sample}.
\end{proof}
\begin{proof}[Proof of Theorem~\ref{thm:distrib_on_Sn_sample} (via Theorem~\ref{thm:distrib_on_Sn})]
Let \smash{$\sigma_0,\sigma_1,\dots,\sigma_M\iidsim q$}, and define the empirical distribution
\[\widehat{q} = \frac{1}{M+1}\sum_{m=0}^M \delta_{\sigma_m},\]
where $\delta_{\sigma}$ is the point mass at $\sigma$.
Now we treat $\widehat{q}$ as fixed. Let $k$ be drawn uniformly from $\{0,\dots,M\}$ (that is, $\sigma_k$ 
is drawn at random from $\widehat{q}$).
Applying Theorem~\ref{thm:distrib_on_Sn} with $\widehat{q}$ in place of $q$, we then see that
\[P = \sum_{\sigma\in\Scal_n} \widehat{q}(\sigma)\cdot \One{T(X_{\sigma\circ\sigma_k^{-1}})\geq T(X)} =  \frac{\sum_{m=0}^M \One{T(X_{\sigma_m\circ\sigma_k^{-1}})  \geq T(X) }}{1+M}\]
is a valid p-value conditional on $\widehat{q}$, and therefore also valid after marginalizing over $\widehat{q}$. Since $\sigma_0,\dots,\sigma_M$ are drawn i.i.d.~and are therefore in a random 
order, we see that
\begin{align*}
P &=  \frac{\sum_{m=0}^M \One{T(X_{\sigma_m\circ\sigma_k^{-1}})  \geq T(X) }}{1+M} \eqd  \frac{\sum_{m=0}^M \One{T(X_{\sigma_m\circ\sigma_0^{-1}})  \geq T(X) }}{1+M}\\ 
&= \frac{1+\sum_{m=1}^M \One{T(X_{\sigma_m\circ\sigma_0^{-1}})  \geq T(X) }}{1+M} ,
\end{align*}
which is the desired p-value.
\end{proof}

\subsection{Another perspective: exchangeable permutations}
Many of the results described above can be viewed through the lens of exchangeability---not on the data $X$ (which we 
assume to be exchangeable under the null hypothesis $H_0$), but on the collection of permutations
used to define the p-value $P$.
\begin{theorem}\label{thm:exch}
Let $\sigma_0,\sigma_1,\dots,\sigma_M\in\Scal_n$ be a random set of permutations, which are exchangeable, i.e.,
\[
(\sigma_0,\sigma_1,\dots,\sigma_M)\eqd (\sigma_{\pi(0)},\sigma_{\pi(1)},\dots,\sigma_{\pi(M)})
\]
for any fixed permutation $\pi$ on $\{0,\dots,M\}$. Then
\[ 
P = \frac{1+\sum_{m=1}^M \One{T(X_{\sigma_m\circ\sigma_0^{-1}}) \geq T(X)}}{1+M} 
\]
is a valid p-value, i.e., $\Pp{H_0}{P\leq \alpha}\leq\alpha$ for all $\alpha\in[0,1]$.
\end{theorem}
\noindent 
Many of the results stated earlier can be viewed as special cases---in
particular, the results for a subgroup $G$, or for a subset $S$, as
well as our more general result Theorem~\ref{thm:distrib_on_Sn_sample}
for permutations drawn i.i.d.~from $q$.

\begin{proof} To be clear, this theorem is essentially just a new
  perspective, and can be proved as a corollary to
  Theorem~\ref{thm:distrib_on_Sn}. To see why, let
  $\sigma_0,\dots,\sigma_M$ be exchangeable, and let
  $\widehat{q} = \frac{1}{M+1}\sum_{m=0}^M\delta_{\sigma_m}$ be the
  empirical distribution induced by the {\em unordered} set of drawn
  permutations.  Then since $\sigma_0,\dots,\sigma_M$ is exchangeable,
  conditional on $\widehat{q}$ it holds that $\sigma_0$ is a random
  draw from $\widehat{q}$. Applying Theorem~\ref{thm:distrib_on_Sn}
  with $\widehat{q}$ in place of $q$ gives the conclusion.
\end{proof} 
However, we can also prove this result in a more intuitive way, using the framework of exchangeability:
\begin{proof}[Alternative proof of Theorem~\ref{thm:exch}]
  Since the sequence $\sigma_0,\sigma_1,\dots,\sigma_M$ is exchangeable,
  \[T(X_{\sigma_0}),T(X_{\sigma_1}),\dots,T(X_{\sigma_M})\]
  is also exchangeable conditional on $X$. It is thus still
  exchangeable after marginalizing over $X$. Therefore, under the null
  hypothesis $H_0$, the test statistic values
\begin{equation}\label{eq:exch-stat}
 T(X)=T(X_{\sigma_0\circ\sigma_0^{-1}}), \ T(X_{\sigma_1\circ\sigma_0^{-1}}),\  \dots, \ T(X_{\sigma_M\circ\sigma_0^{-1}})
\end{equation}
are also exchangeable---this follows immediately from the previous
line because $X \eqd  X_{\sigma_0^{-1}}$ under $H_0$.  This shows that the
p-value $P$ defined in~\eqref{eqn:pvalue_general_distribution_sample}
is valid.
\end{proof}

\section{Averaging to reduce variance}
The p-value $P$ defined in~\eqref{eqn:pvalue_general_distribution} can equivalently be written as
\[P = \Ppst{\sigma \sim q}{T(X_{\sigma\circ\sigma_0^{-1}})\geq
  T(X)}{X,\sigma_0}.\]
It is clear that $P$ is random even if we condition on the observed
data $X$, because of the randomness due to $\sigma_0$.  Consequently,
in some settings $P$ may be quite variable conditional on the data
$X$, and this may be undesirable.

To address this issue, we can also consider averaging over $\sigma_0$
(in addition to averaging over $\sigma$) in the calculation of
$P$. This alternative definition is now a deterministic function of
the observed data $X$, but may no longer be a valid
p-value. Nonetheless, the following theorem shows a bound on the Type
I error.
\begin{theorem}\label{thm:factor_of_2}
Let $q$ be any distribution over $\sigma\in\Scal_n$. Define
\begin{equation}\label{eqn:pvalue_general_distribution_factor_of_2}
\bar P = \sum_{\sigma,\sigma_0\in\Scal_n} q(\sigma)q(\sigma_0) \cdot\One{T(X_{\sigma\circ\sigma_0^{-1}})\geq T(X)},\end{equation}
or equivalently,
\[\bar P = \Ppst{\sigma,\sigma_0 \iidsim q}{T(X_{\sigma\circ\sigma_0^{-1}})\geq T(X)}{X}.\]
Then $\bar P$ is a valid p-value up to a factor of 2, 
i.e., $\Pp{H_0}{\bar P\leq \alpha}\leq 2\alpha$ for all $\alpha\in[0,1]$.
 In other words, the quantity $\min\{2\bar P,1\}$ is a valid p-value.
\end{theorem}
\begin{proof}
Draw \smash{$\sigma_0^{(1)},\sigma_0^{(2)},\dots\iidsim q$}. Let 
\[P_m= \sum_{\sigma\in\Scal_n} q(\sigma)\cdot\One{T(X_{\sigma\circ\sigma_0^{(m)}{}^{-1}})\geq T(X)},\]
for each $m\geq1$. Then by Theorem~\ref{thm:distrib_on_Sn},
each $P_m$ is a valid p-value. It is known \citep{ruschendorf1982random,vovk2020combining} that the average of valid
p-values is a valid up to a factor of 2, i.e., for any $M\geq 1$ the average $\bar{P}_M = \frac{1}{M}\sum_{m=1}^M P_m$
satisfies $\PP{\bar{P}_M\leq \alpha}\leq 2\alpha$ for all $\alpha\in[0,1]$. We can equivalently write
\[\bar{P}_M
= \sum_{\sigma'\in \Scal_n} \frac{\sum_{m=1}^M \One{\sigma_0^{(m)}=\sigma'}}{M} 
\cdot \sum_{\sigma\in\Scal_n} q(\sigma)\cdot\One{T(X_{\sigma\circ\sigma'{}^{-1}})\geq T(X)}.\]
By the Law of Large Numbers, $\bar{P}_M$ converges almost surely to the p-value $\bar P$
defined in~\eqref{eqn:pvalue_general_distribution_factor_of_2},
which completes the proof.
\end{proof}

Returning to Example~\ref{example1}, we see that while $P$ was a
\emph{mixture} of
$P_{\mathrm{Id}}, P_{\sigma_{1\leftrightarrow 4, 2\leftrightarrow 3}},
P_{\sigma_{1\leftrightarrow 3, 2\leftrightarrow 4}}$,
we now have that $\bar P$ is an \emph{average} of these, meaning
$\bar P = \tfrac13(P_{\mathrm{Id}} + P_{\sigma_{1\leftrightarrow 4,
    2\leftrightarrow 3}} + P_{\sigma_{1\leftrightarrow 3,
    2\leftrightarrow 4}})$. Simplifying, we get
\[
\bar P = \begin{cases}
\frac{5}{9}, & \textnormal{ w.p. } 1/2, \\
1, & \textnormal{ w.p. } 1/2.
\end{cases}
\]
It is worth noting that this new quantity $\bar P$ is neither more conservative nor more anti-conservative than the p-value $P$~\eqref{eq:example-fixed-P} from earlier. This is perhaps a more general phenomenon: the average of p-values need not in general be anti-conservative, and indeed it could often be more conservative, than the original p-values.

Analogously, the p-value in Theorem~\ref{thm:distrib_on_Sn_sample}, computed via random samples from $q$, can also be averaged to reduce variance.
\begin{theorem}\label{thm:factor_of_2_sample}
Let $q$ be any distribution over $\sigma\in\Scal_n$. 
Let \smash{$\sigma_0,\sigma_1,\dots,\sigma_M\iidsim q$}, and define
\begin{equation}\label{eqn:pvalue_general_distribution_sample}
\bar{P} = \frac{\sum_{m=0}^M\sum_{m'=0}^M \One{T(X_{\sigma_m\circ\sigma_{m'}^{-1}})  \geq T(X) }}{(1+M)^2}.\end{equation}
Then $P$ is a valid p-value up to a factor of 2,
 i.e., 
$\Pp{H_0}{P\leq \alpha}\leq 2\alpha$ for all $\alpha\in[0,1]$.
 Thus, as before,  the quantity $\min\{2\bar P,1\}$ is a valid p-value.
\end{theorem}
\noindent The proof is similar to that of
Theorem~\ref{thm:factor_of_2}, and we omit it for brevity.

\section{Connections to the literature}\label{sec:connections}

We next mention a few connections to  the broader literature.

\subsection{Permutation tests vs randomization tests}\label{sec:tea}
\citet{hemerik2021another} describe the difference between two testing frameworks, 
permutation tests (as studied in our present work) versus randomization tests.
The difference is subtle, because randomization tests may still use permutations.
Specifically,
\citet{hemerik2021another} highlight
\begin{quote}
an important difference in mathematical reasoning between
these classes: a permutation test fundamentally requires that the set of permutations has a group
structure, in the algebraic sense; the reasoning behind a randomisation test is not based on such
a group structure, and it is possible to use an experimental design that does not correspond to a
group.
\end{quote}

To better understand this distinction, we can consider a scenario where a fixed subset $S\subseteq\Scal_n$,
which is not a subgroup, is used for a randomization test rather than a permutation test.
Consider a study comparing a treatment versus a placebo, with $n/2$ many subjects assigned to 
each of the two groups. We can use a permutation $\sigma$ to denote the treatment assignments,
with $\sigma(i)\leq n/2$ indicating that subject $i$ receives the treatment, and $\sigma(i)>n/2$ indicating
that subject $i$ receives the placebo. 
Now we switch notation, to be able to compare to permutation tests more directly---writing
$X= (1,\dots,1,0,\dots,0)$, suppose
that we will assign treatments via the permuted vector $X_{\sigma}$, i.e., for each subject
$i=1,\dots,n$, under this permutation $\sigma$
 the $i$th subject will receive the treatment if $X_{\sigma(i)}=1$,
or the placebo if $X_{\sigma(i)}=0$.

Now suppose that we draw a random treatment assignment $\sigma_{\textnormal{asgn}}\sim\textnormal{Unif}(S)$,
from a fixed subset $S\subseteq\Scal_n$ (for example, $S$ may be chosen to restrict to treatment assignments
that are equally balanced across certain subpopulations). After the treatments are administered,
the measured response variable
is given by $Y=(Y_1,\dots,Y_n)$. Fix any test statistic $T(X) = T(X,Y)$ (we will implicitly condition on $Y$),
and compute
\begin{equation}\label{eqn:pvalue_subset_randomization_test}
P = \frac{\sum_{\sigma\in S}\One{T(X_{\sigma})\geq T(X_{\sigma_{\textnormal{asgn}}})}}{|S|}.\end{equation}
Since $\sigma_{\textnormal{asgn}}$ was drawn uniformly from $S$, this quantity $P$ is a valid p-value. 
In the terminology of~\citet{hemerik2021another}, this test is a randomization test,
not a permutation test. While the set
of possible
treatment assigments $\{X_\sigma: \sigma\in S\}$  happens to be indexed by permutations $\sigma$,
the group structure of permutations is not used in any way, and we do not rely on any invariance properties.

Comparing to the invalid p-value $P = \frac{\sum_{\sigma \in S} \One{T(X_\sigma)\geq T(X)}}{|S|}$
considered in~\eqref{eqn:pvalue_subset}, we can easily see the distinction:
for a randomization test, the observed statistic is $T(X_{\sigma_{\textnormal{asgn}}})$ for
a randomly drawn $\sigma_{\textnormal{asgn}}\sim\textnormal{Unif}(S)$, 
while in the permutation test in~\eqref{eqn:pvalue_subset}, the observed statistic is $T(X)$
(i.e., using the {\em fixed} permutation $\textnormal{Id}$ in place of a randomly drawn $\sigma_{\textnormal{asgn}}$).
For this reason, the randomization test p-value in~\eqref{eqn:pvalue_subset_randomization_test}
is valid, while the permutation test calculation in~\eqref{eqn:pvalue_subset} is not valid in general.

Now we again consider \citet{hemerik2018exact}'s method using a fixed subset.
This test~\eqref{eqn:pvalue_subset_corrected} is a permutation test, not a randomization test---the observed data $X$, and its
corresponding statistic $T(X)$, do not arise from a random treatment assignment.
More generally, our proposed test~\eqref{eqn:pvalue_general_distribution} using an arbitrary
distribution $q$ on $\Scal_n$ is again a permutation test rather than a randomization test---that is, the observed
data is given by $X$ itself, not by a randomly chosen treatment assignment $X_{\sigma_{\textnormal{asgn}}}$
for $\sigma_{\textnormal{asgn}}\sim q$.
Nonetheless, we are able to produce a valid p-value without assuming an underlying group structure
or uniform sampling for the permutations considered by the test.

\subsection{Exchangeable MCMC}\label{sec:mcmc}

The result of Theorem~\ref{thm:distrib_on_Sn_sample}, which allows for random samples drawn from an arbitrary
distribution $q$ on $\Scal_n$, is closely connected to \citet{besag1989generalized}'s well known construction for 
obtaining exchangeable samples from Markov chain Monte Carlo (MCMC) sampling.

Consider a distribution $Q_0$ on $\Zcal$, and suppose we want to test 
\[H_0: \ Z\sim Q_0\]
with some test statistic $T(Z)$. To find a significance threshold for
$T(Z)$, we would ideally like to draw from the null distribution,
i.e., compare $T(Z)$ against $T(Z_1),\dots,T(Z_M)$ for
\smash{$Z_1,\dots,Z_m \iidsim Q_0$}.  However, in many settings, sampling
directly from $Q_0$ is impossible, but we instead have access to a Markov chain
whose stationary distribution is $Q_0$. If we run the Markov chain 
initialized at $Z$ to obtain draws $Z_1,\dots,Z_M$ (say, running the
Markov chain for some fixed number of steps $s$ between each draw),
then dependence among these sequentially drawn samples means that
$Z,Z_1,\dots,Z_M$ are not i.i.d., and are not even
exchangeable. Without studying the mixing properties of the
Markov chain, we cannot determine how large the number of steps needs
to be for the dependence to become negligible. Instead,
\citet{besag1989generalized} propose a construction where the samples
are drawn in parallel (rather than sequentially), which ensures
exchangeability:
\begin{theorem}[{\citet[Section 2]{besag1989generalized}}]\label{thm:besag}
  Let $Q_0$ be any distribution on a probability space
  $\Zcal$. Construct a Markov chain on $\Zcal$ with stationary
  distribution $Q_0$, whose forward and backward transition
  distributions (initialized at $z\in\Zcal$) are denoted by
  $Q_{\rightarrow}(\cdot|z)$ and $Q_{\leftarrow}(\cdot|z)$. Let
  $Q_{\rightarrow}^s(\cdot|z)$ and $Q_{\leftarrow}^s(\cdot|z)$ denote
  the forward and backward transition distributions after running $s$
  steps of the Markov chain, for some fixed $s\geq 1$.  
  Given an
  initialization $Z$, suppose we generate data as in the left plot of Figure~\ref{fig:BC}:
\[\begin{cases}\textnormal{First, draw $Z_* \sim Q^s_{\leftarrow}(\cdot|Z)$;}\\
\textnormal{Then, draw $Z_1,\dots,Z_M\iidsim Q^s_{\rightarrow}(\cdot|Z_*)$.}\end{cases}\]
If it holds marginally that $Z\sim Q_0$, then the draws $Z,Z_1,\dots,Z_M$ are exchangeable.
\end{theorem}
\noindent Given this exchangeability property, the quantity
$P = \frac{1 + \sum_{m=1}^M \One{T(Z_m)\geq T(Z)}}{1+M}$ is then a
valid p-value for testing $H_0: Z\sim Q_0$.  

\begin{figure}[t] 
\begin{tikzpicture}

\begin{scope}[xshift=0cm,yshift=0cm]

\node (X) at (-3,0) {};
\node (Xhub) at (0,0) {};
\node (X1) at (0.75,2.9) {};
\node (X2) at (2.7,1.3) {};
\node (Xdots) at (2.7,-1.3) {};
\node (XM) at (0.75,-2.9) {};
\node (Xtext) at (-3,0) {$Z$};
\node (Xhubtext) at (0,0) {$Z_*$};
\node (X1text) at (0.75, 2.9) {$Z_1$};
\node (X2text) at (2.7,1.3) {$Z_2$};
\node (Xdotstext) at (2.7,-1.3) {$\vdots$};
\node (XMtext) at (0.75,-2.9) {$Z_M$};
\node (Xhubstep) at (-1.5,0.2) {\scriptsize backward step};
\node (Xhubstep_) at (-1.5,-0.2) {\scriptsize via $Q_{\leftarrow}(\cdot |Z)$};
\node (X1step) at (-0.65,1.75) {\scriptsize forward step};
\node (X1step_) at (-0.65,1.35) {\scriptsize via $Q_{\rightarrow}(\cdot |Z_*)$};
\node (X2step) at (2.45,0.55) {\scriptsize forward step};
\node (X2step_) at (2.45,0.15) {\scriptsize via $Q_{\rightarrow}(\cdot |Z_*)$};
\node (XMstep) at (-0.65,-1.35) {\scriptsize forward step};
\node (XMstep_) at (-0.65,-1.75) {\scriptsize via $Q_{\rightarrow}(\cdot |Z_*)$};
\draw[->][line width=0.5mm,shorten >=5pt,shorten <=5pt] (X.east) -- (Xhub.west);
\draw[<-][line width=0.5mm,shorten >=5pt,shorten <=5pt] (X1.south east) -- (Xhub.north west);
\draw[<-][line width=0.5mm,shorten >=5pt,shorten <=5pt] (X2.south) -- (Xhub.north);
\draw[<-][line width=0.5mm,shorten >=5pt,shorten <=5pt] (Xdots.north) -- (Xhub.south);
\draw[<-][line width=0.5mm,shorten >=5pt,shorten <=5pt] (XM.north east) -- (Xhub.south west);
\end{scope}

\begin{scope}[xshift=8.5cm,yshift=0cm]

\node (X) at (-3,0) {};
\node (Xhub) at (0,0) {};
\node (X1) at (0.75,2.9) {};
\node (X2) at (2.7,1.3) {};
\node (Xdots) at (2.7,-1.3) {};
\node (XM) at (0.75,-2.9) {};
\node (Xtext) at (-3,0) {$X$};
\node (Xhubtext) at (0,0) {$X_{\sigma_0^{-1}}$};
\node (X1text) at (0.75, 2.9) {$X_{\sigma_1\circ\sigma_0^{-1}}$};
\node (X2text) at (2.7,1.3) {$X_{\sigma_2\circ\sigma_0^{-1}}$};
\node (Xdotstext) at (2.7,-1.3) {$\vdots$};
\node (XMtext) at (0.75,-2.9) {$X_{\sigma_M\circ\sigma_0^{-1}}$};
\node (Xhubstep) at (-1.5,0.2) {\scriptsize draw $\sigma_0\sim q$};
\node (Xhubstep_) at (-1.5,-0.2) {\scriptsize \& apply $\sigma_0^{-1}$};
\node (X1step) at (-0.65,1.75) {\scriptsize draw $\sigma_1\sim q$};
\node (X1step_) at (-0.65,1.35) {\scriptsize \& apply $\sigma_1$};
\node (X2step) at (2.45,0.55) {\scriptsize draw $\sigma_2\sim q$};
\node (X2step_) at (2.45,0.15) {\scriptsize \& apply $\sigma_2$};
\node (XMstep) at (-0.65,-1.35) {\scriptsize draw $\sigma_M\sim q$};
\node (XMstep_) at (-0.65,-1.75) {\scriptsize \& apply $\sigma_M$};
\draw[->][line width=0.5mm,shorten >=8pt,shorten <=5pt] (X.east) -- (Xhub.west);
\draw[<-][line width=0.5mm,shorten >=5pt,shorten <=5pt] (X1.south east) -- (Xhub.north west);
\draw[<-][line width=0.5mm,shorten >=5pt,shorten <=8pt] (X2.south) -- (Xhub.north);
\draw[<-][line width=0.5mm,shorten >=8pt,shorten <=5pt] (Xdots.north) -- (Xhub.south);
\draw[<-][line width=0.5mm,shorten >=5pt,shorten <=5pt] (XM.north east) -- (Xhub.south west);
\end{scope}

\end{tikzpicture}
\caption{Left: \citet{besag1989generalized}'s parallel construction (with $s=1$). Right: the construction used in Theorem~\ref{thm:distrib_on_Sn_sample}.}
\label{fig:BC}
\end{figure}
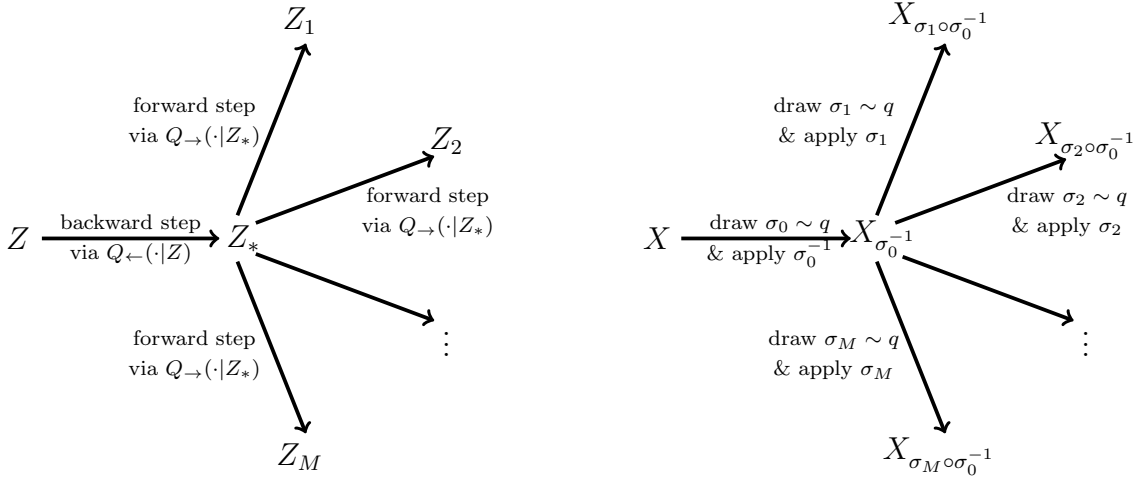

Now we will see how Theorem~\ref{thm:distrib_on_Sn_sample} is related to this
result. Let $\Zcal = \Xcal^n$, and let $Q_0$ be any exchangeable
distribution. In the setting of this paper, we do not know $Q_0$
precisely, which makes it a bit different from a typical setting
where~\cite{besag1989generalized}'s method is applied. However, we
will work with a Markov chain for which {\em any} exchangeable
distribution $Q_0$ is stationary, and in fact, Theorem~\ref{thm:besag} holds
 regardless
of whether $Q_0$ is the unique stationary distribution for the Markov
chain.

Consider the Markov chain given by applying a randomly chosen permutation $\sigma\sim q$, that is,
for $x = (x_1,\dots,x_n)$,
\[
Q_{\rightarrow}(\cdot|x) = \sum_{\sigma \in S} q(\sigma) \cdot\delta_{x_\sigma},
\]
where $\delta_{x_\sigma}$ is the point mass at $x_\sigma$, while the backward transition probabilities are given by
\[
Q_{\leftarrow}(\cdot|x) = \sum_{\sigma \in S}q(\sigma) \cdot\delta_{x_{\sigma^{-1}}}.
\]
Then, to implement the test described in Theorem~\ref{thm:distrib_on_Sn_sample}, we run \citet{besag1989generalized}'s method (with $s=1$): we define $X_* = X_{\sigma_0^{-1}}$,
and then define $X_m = (X_*)_{\sigma_m} = X_{\sigma_m\circ\sigma_0^{-1}}$ for $m=1,\dots,M$. This is illustrated on the right-hand side
of Figure~\ref{fig:BC}. If $X$ is exchangeable (that is, it is drawn from some exchangeable $Q_0$), then
the exchangeability of $X,X_1,\dots,X_M$ follows
by Theorem~\ref{thm:besag}, and this verifies that $P$ is a valid p-value, thus completing the proof of Theorem~\ref{thm:distrib_on_Sn_sample}.

Of course, we have only written out our method for the $s=1$ case
(where $s$ is the number of steps of the Markov chain). New variants
of our method can be constructed by taking $s > 1$ backward steps to
the hidden node, and the same number $s$ of forward steps to the permuted
data. All of these are valid for the same reason as the $s=1$ case.

\section{Conclusion}\label{sec:conclusion}

We proposed a new method for permutation testing that generalizes
previous methods.  This idea naturally opens up new lines of
theoretical and practical enquiry.  In this work, we have focused on
validity, but it is of course also important to examine the
consistency and power of such methods. 
In particular,
\citet{dobriban2021consistency,kim2022minimax} study
the power of the permutation test when using the full permutation group $\Scal_n$;
it would be interesting to examine this question in the context of
 using only a subset $S\subseteq\Scal_n$ or a nonuniform
distribution over $\Scal_n$.
In addition, the theoretical guarantees for all the permutation tests
considered here ensure a p-value $P$ that is valid in the sense of
satisfying $\Pp{H_0}{P\leq \alpha}\leq\alpha$, which means that $P$
could potentially be quite conservative under the null (for instance,
we saw this behavior when `fixing' the failure example in Section
\ref{sec:testing}). It would also be interesting to understand which
types of tests reduce overly conservative outcomes.

In conclusion, it is perhaps remarkable that one can still gain new understanding about classical permutation methods. In turn, 
this enhanced understanding can inform other areas of inference. As an example, the results from this paper were
 motivated by questions in conformal prediction \citep{vovk2005algorithmic}, a method for distribution-free predictive inference.
Classically, conformal prediction has relied on exchangeability of data points (e.g., training and test data are drawn i.i.d.~from
 the same unknown distribution),
 and thus the  joint distribution of the data (including both training samples and  a test point) is invariant under an arbitrary permutation.
In contrast, in our recent work \citep{barber2022conformal},
 we studied the problem of constructing prediction intervals when the data do not satisfy exchangeability; for instance,
  the distribution of observations may simply drift over time in an unknown fashion. Thus the data is no longer invariant under
  an arbitrary permutation, and so we instead restrict attention to 
  a weighted distribution over simple permutations that 
  only swap the test point with a random training point, which at least approximately preserve the distribution of the data. These swaps clearly do not form a subgroup of permutations, and are weighted non-uniformly;
 understanding how permutation tests operate in this setting, as in Theorem~\ref{thm:distrib_on_Sn}, is key to the findings in our aforementioned work.

\subsection*{Conflicts of Interest}
The authors have no conflict of interest to declare.

\subsection*{Acknowledgments}
The authors thank Nick Koning and Ilmun Kim for helpful feedback on an early preprint. The authors also thank the SQUARE program run by the American Institute of Mathematics, where our collaboration started.
 R.F.B.~was 
supported by the National Science Foundation via grants DMS-1654076 and DMS-2023109,
and by the Office of Naval Research via grant N00014-20-1-2337. 
E.J.C.~was supported by the Office of Naval Research grant N00014-20-1-2157, the National Science Foundation grant DMS-2032014, the Simons Foundation under award 814641, and the ARO grant 2003514594. R.J.T.~was supported by ONR grant N00014-20-1-2787.

\bibliographystyle{plainnat}
\bibliography{bib}
\end{document}